\pdfoutput=1
\RequirePackage{ifpdf}
\ifpdf 
\documentclass[pdftex]{sigma}
\else
\documentclass{sigma}
\fi

\numberwithin{equation}{section}

\newtheorem{Theorem}{Theorem}[section]

\newtheorem{Lemma}[Theorem]{Lemma}
\newtheorem{Proposition}[Theorem]{Proposition}
 { \theoremstyle{definition}

 }

\usepackage{mathrsfs}

\newcommand{\CC}{\mathbb{C}}
\newcommand{\II}{\mathbb{I}}
\newcommand{\cK}{\mathscr{K}}
\newcommand{\cR}{\mathscr{R}}
\newcommand{\cL}{\mathscr{L}}

\begin{document}

\allowdisplaybreaks

\newcommand{\arXivNumber}{1710.08490}

\renewcommand{\PaperNumber}{094}

\FirstPageHeading

\ShortArticleName{Algebraic Bethe Ansatz for the XXZ Gaudin Models with Generic Boundary}

\ArticleName{Algebraic Bethe Ansatz for the XXZ Gaudin Models\\ with Generic Boundary}

\Author{Nicolas CRAMPE}

\AuthorNameForHeading{N.~Crampe}

\Address{Laboratoire Charles Coulomb (L2C), UMR 5221 CNRS-Universit{\'e} de Montpellier,\\
Montpellier, France}
\Email{\href{mailto:nicolas.crampe@umontpellier.fr}{nicolas.crampe@umontpellier.fr}} 

\ArticleDates{Received November 01, 2017, in f\/inal form December 06, 2017; Published online December 13, 2017}

\Abstract{We solve the XXZ Gaudin model with generic boundary using the modif\/ied algebraic Bethe ansatz.
 The diagonal and triangular cases have been recovered in this general framework. We show that the model
 for odd or even lengths has two dif\/ferent behaviors. The corresponding Bethe equations are computed for all the cases.
 For the chain with even length, inhomogeneous Bethe equations are necessary. The higher spin Gaudin models with generic
 boundary is also treated.}


\Classification{81R12; 17B80; 37J35}

\section{Introduction}

The algebraic Bethe ansatz is a powerful method
to solve analytically numerous integrable models \cite{TF}.
Then, the introduction of integrable boundaries in this framework was initiated in \cite{sklyanin}.
However, the problem to solve models with generic integrable boundaries thanks to the algebraic Bethe ansatz has been
overcome only recently \cite{Bel++,Bel,BC,Bel+,cra}.
The problem lies in the fact that there is no simple particular eigenvector, which is normally the starting point of the method.
Therefore, other approaches have been used to deal with this problem:
the functional Bethe ansatz \cite{FraGSW11,Gal08,MurN05, nepo},
the coordinate Bethe ansatz \cite{CRS1+,CRS1}, the separation of variables \cite{niccoli2+,FSW,niccoli2},
the $q$-Onsager approach \cite{BK+,BK} and the matrix ansatz \cite{CRS2, GLMV1+,GLMV1++,GLMV1,GLMV1+++}.
The algebraic Bethe ansatz has been also used in this context \cite{BCR13,CLSW,MMR,pimenta,YanZ07} but constraints on the parameters def\/ining the boundaries
were necessary.
Then, inhomogeneous T-Q relations were introduced in \cite{CYSW1,CYSW1+,CYSW2,Nep} where the authors obtained the eigenvalues and the Bethe equations
for the XXZ model with generic boundaries. Finally, the modif\/ied algebraic Bethe ansatz allows us to compute the associated eigenvectors \cite{Bel++,Bel,BC,Bel+,ZLCYSW}.
Let us mention also that this method has been also exploited to study twisted XXX spin chain in~\cite{BP}
or the rational Gaudin magnets in arbitrarily oriented magnetic f\/ields in \cite{FT}.

In this letter, we solve the Gaudin model (introduced in \cite{gau}) with generic boundaries using the modif\/ied algebraic Bethe ansatz.
The Gaudin model is one of the simplest quantum integrable system and we hope that its resolution will shed more light on the method.
Let us also mention that the algebraic Bethe ansatz was already applied to solve open
Gaudin model using the vertex-IRF correspondence of the XXZ spin chain \cite{HYF,YSZ}. However, as pointed in \cite{CMN}, the model
considered in \cite{HYF,YSZ} is obtained as a limit from the XXZ spin chain and is dif\/ferent of the one studied here which is constructed directly from
a classical $r$-matrix.

This letter is organised as follows. In Section~\ref{sec:model}, we recall well-known results about the Gaudin algebra and model to f\/ix the notations
used in the following. Then, we provide technical results in the Sections~\ref{sec:alg} and~\ref{sec:rep} to implement the algebraic Bethe ansatz.
In Section~\ref{sec:tria}, we recover the results for the diagonal and triangular boundaries \cite{Hik,mano+,mano}. Section~\ref{sec:odd} contains new results for the generic boundaries when the length of the chain is odd. Finally, we deal with the even chain in Section~\ref{sec:even} where
an additional computation is necessary. We explain also why the computation is dif\/ferent depending on the parity of the length of the chain.
In Section~\ref{sec:HS}, we generalize these results to solve the higher spin Gaudin models with generic boundary.

\section{Gaudin models and algebras} \label{sec:model}

In this section, we recall dif\/ferent well-known notions to construct Gaudin models and prove their integrability.

For a matrix $r(x,y)\in \operatorname{End}(\CC^N\otimes \CC^N)$ depending on 2 parameters $x,y$, one def\/ines the fundamental equation, called classical Yang--Baxter equation, given by
 \begin{gather*}
 [{r}_{13}(x_1,x_3) , {r}_{23}(x_2,x_3)]=[{r}_{21}(x_2,x_1),{r}_{13}(x_1,x_3)]+[{r}_{23}(x_2,x_3),{r}_{12}(x_1,x_2)].
 \end{gather*}
 In the previous relation, we have used the usual notations: $r_{12}(x) = r(x)\otimes \II$, $r_{23}(x) =\II\otimes r(x) $, \dots,
 where $\II$ is the identity matrix. A solution of the classical Yang--Baxter equation is called a $r$-matrix.
 In addition, if an $r$-matrix $r(x,y)$ depends only on the quotient $x/y$ and satisf\/ies the supplementary relations $r_{12}(x,y)= r_{12}(x/y)=-r_{21}(y/x)=-r_{12}(y/x)^{t_1t_2}$
 then it is called a skew-symmetric $r$-matrix.
 For a given skew-symmetric $r$-matrix $r(x/y)$, one def\/ines also the ref\/lection equation
 \begin{gather*}
 r_{12}(x/y) k_1(x) k_2(y)-k_1(x) k_2(y)r_{21}(x/y)= k_2(y)r_{12}(xy)k_1(x) - k_1(x) r_{21}(xy) k_2(y).
 \end{gather*}
A solution $k(x)$ of the ref\/lection equation is called $k$-matrix.
 From a skew-symmetric $r$-matrix and an associated $k$-matrix, we can construct a new $r$-matrix \cite{Skr}
 \begin{gather}\label{def:rb}
 \overline{r}_{12}(x,y)=r_{12}(x/y)-k_1(x) r_{12}(1/(xy)) k_1(x)^{-1} ,
 \end{gather}
 which is not skew-symmetric.

 To each $r$-matrix $r(x,y)$, on can associate a Lie algebra $\cR$.
 Indeed, if ${\cK}(x)\in \operatorname{End}\big(\CC^N\big) \otimes \cR$, then the following relation
\begin{gather}\label{eq:Al}
 [ {\cK}_{1}(x) , {\cK}_{2}(y) ]=[ {r}_{21}(y,x) , {\cK}_{1}(x) ]+[ {\cK}_{2}(y) , {r}_{12}(x,y) ]
\end{gather}
def\/ines a Lie commutator for $\cR$. Indeed, the antisymmetry of the product is obvious and the Jacobi identity is satisf\/ied due to the classical Yang--Baxter equation. For any invertible matrix $M(x)\in \operatorname{End}\big(\CC^N\big)$, there is a Lie algebra isomorphism given by
\begin{align*}
 \Phi_M\colon \quad \quad \cR&\rightarrow \cR_M, \\
 {\cK_0}(x)&\mapsto M_0(x)^{-1} \cK_0(x) M_0(x),
\end{align*}
where $\cR_M$ is the Lie algebra def\/ined by \eqref{eq:Al} with the $r$-matrix
\begin{gather*}
 M_1(x)^{-1}M_2(y)^{-1} r_{12}(x,y)M_1(x)M_2(y).
\end{gather*}
From the def\/ining relations \eqref{eq:Al}, one can show that the transfer matrix
\begin{gather}\label{eq:tr}
 t(x)=\frac{1}{4}\operatorname{tr}_0\big({\cK_0}(x)^2\big)
\end{gather}
satisf\/ies $ [t(x),t(y)]=0$. Then, when a representation for $\cR$ is chosen, the coef\/f\/icients of $t(x)$ provide an integrable hierarchy.
Let us remark that, for a given matrix $M(x)$, $\cR$ and $\cR_M$ give the same hierarchy since it is easy to see that $\operatorname{tr}\big({\cK_0}(x)^2\big)= \operatorname{tr}\big({\Phi_M(\cK_0}(x))^2\big)$.

In this letter, we focus on the skew-symmetric $r$-matrix associated to the af\/f\/ine Kac--Moody algebra $\widehat{{\mathfrak{sl}}_2}$:
\begin{gather}\label{def:r}
 r(x)=\frac{1}{x-1}\begin{pmatrix}
 -\frac{1}{2}(x+1)&0&0&0\\
 0&\frac{1}{2}(x+1)& -2&0\\
 0&-2x& \frac{1}{2}(x+1) &0\\
 0&0&0&-\frac{1}{2}(x+1)
 \end{pmatrix}.
\end{gather}
The most general $k$-matrix associated to this $r$-matrix is given by
\begin{gather}\label{eq:k2}
 k(x)=\begin{pmatrix}
 \beta+\gamma/x& - \frac{\alpha(\beta+\rho)}{2}(x-1/x)\\
 \frac{\beta-\rho}{2\alpha}(x-1/x)&\beta+\gamma x&
 \end{pmatrix},
\end{gather}
where $\alpha$, $\beta$, $\gamma$ and $\rho$ are scalar parameters. As explained previously,
the $r$-matrix and the $k$-matrix allow us to construct a new $r$-matrix (see relation~\eqref{def:rb}).
Without loss of generality, we can use the Lie algebra isomorphism $\Phi_M$ to transform this $r$-matrix.
Therefore, we use in this paper the following $r$-matrix
\begin{gather}\label{eq:rt}
\widetilde{r}_{12}(x,y)=M_1(x)^{-1}M_2(y)^{-1}\big( r_{12}(x/y)-k_1(x) r_{12}(1/(xy)) k_1(x)^{-1} \big)M_1(x)M_2(y),
\end{gather}
where $r(x)$ is given by \eqref{def:r}, $k(x)$ is given by \eqref{eq:k2} and
\begin{gather*}
 M(x)=\begin{pmatrix}
 \frac{\beta+\rho}{2\rho} & \frac{\alpha}{x}\\
 \frac{(\beta-\rho)x}{2\alpha\rho} & 1
 \end{pmatrix}.
\end{gather*}
Explicitly, this $r$-matrix reads as
\begin{gather*}
 \widetilde{r}(x,y)=\begin{pmatrix}
 -\frac{1}{2}\omega(x,y) \sigma^z&b(x) \sigma^z+ f(y,x) \sigma^-\\
 c(x) \sigma^z
 - f(1/y,1/x) \sigma^+&\frac{1}{2}\omega(x,y)\sigma^z
 \end{pmatrix},
\end{gather*}
where $\sigma^z=\begin{pmatrix}
 1&0\\
 0&-1
 \end{pmatrix}$, $\sigma^+=\begin{pmatrix}
 0&1\\
 0&0
 \end{pmatrix}$ and $\sigma^-=\begin{pmatrix}
 0&0\\
 1&0
 \end{pmatrix}$ and
 \begin{gather}
 \delta(x)=\frac{( (\beta-\rho)x^2 + 2\gamma x +\rho + \beta)}{2\rho (x-1/x)},\qquad b(x)=-\frac{\alpha}{\delta(x)} , \qquad
c(x)=\frac{\rho^2-\beta^2}{4\alpha\rho^2 \delta(1/x)},\label{eq:defbc} \\
 \omega(x,y)=\frac{x+y}{x-y}+\frac{xy+1}{xy-1} \qquad\text{and}\qquad f(x,y)=\omega(x,y) \frac{\delta(x)}{\delta(y)}.\label{eq:defof}
 \end{gather}
 Let us emphasize that the $r$-matrix $\widetilde r(x,y)$ has a simpler form than the one without the isomorphism
 (i.e., $r_{12}(x/y)-k_1(x) r_{12}(1/(xy)) k_1(x)^{-1}$). Without the isomorphism, the 16 entries of the $r$-matrix should be non-zero whereas
 $\widetilde r(x,y)$ has 6 vanishing entries and each two by two block submatrices are diagonal or triangular.

As explained previously, the $r$-matrix $\widetilde r(x,y)$ allows one to def\/ine a Lie algebra by using relation \eqref{eq:Al} with
\begin{gather}\label{eq:rep0}
 \widetilde{\cK_0}(x)= \begin{pmatrix}
 A(x) & \widetilde B(x) \\
 \widetilde C(x) & -A(x)
 \end{pmatrix} .
\end{gather}
We denote by $\widetilde{\cR}$ this Lie algebra.
Let us mention that the Lie algebra $\widetilde{\cR}$ for special choice of the parameters in the $k$-matrix has been
identif\/ied as Onsager algebras \cite{BBC17}\footnote{To compare this letter with \cite{BBC17}, we must use the additional property
$r_{12}(x)^{t_1}=\sigma^y_1r_{12}(x)\sigma^y_1$ satisf\/ied by the $r$-matrix \eqref{def:r}.}.

Finally, let us introduce the following representation\footnote{We keep the same notation $\widetilde{\cK}(x)$
for the algebraic element or its representation.}
\begin{gather}\label{eq:rep}
 \widetilde{\cK}_0(x)=\sum_{j=1}^L \widetilde{r}_{0j}(x,v_j).
\end{gather}
By using the classical Yang--Baxter equation, it is obvious to show that \eqref{eq:rep} satisf\/ies the commutation relations of $\widetilde{\cR}$.
The parameters $v_i$ are called inhomogeneous parameters.
The Hamiltonians are usually def\/ined by
\begin{gather*}
 \widetilde{H}_j= \operatorname{Res}_{x=v_j} t(x)=\frac{1}{4}\operatorname{Res}_{x=v_j}\operatorname{tr}_0\big({\widetilde{\cK_0}}(x)^2\big)\\
 \hphantom{\widetilde{H}_j}{} =-v_j \sum_{\genfrac{}{}{0pt}{1}{p=1}{p\neq j}}^L \widetilde{r}_{jp}(v_j,v_p)
 +Lv_j\left(\frac{3v_j^2}{v_j^2-1}-\nu(v_j)\right), \label{eq:ham}
\end{gather*}
where
\begin{gather}
 \nu(z)=\frac{z^2(\rho^2-\beta^2)-2\beta\gamma z+\rho^2-\beta^2}{2\rho^2(z^2-1)\delta(z)\delta(1/z)}.\label{eq:nu}
\end{gather}
These Hamiltonians are integrable due to the commutation of the transfer matrix and $[\widetilde{H}_j,\widetilde{H}_k]$ $=0$. They are related to the Gaudin models with boundary (or also called BC Gaudin models). Indeed, by conjugating the previous Hamilonians by $M_1(v_1)\cdots M_L(v_L)$, one gets the Hamiltonians
\begin{gather*}
 H_j=-v_j \sum_{\genfrac{}{}{0pt}{1}{p=1}{p\neq j}}^L \left( r_{jp}\left(\frac{v_j}{v_p}\right) + k_j(v_j) r_{jp}\left(\frac{1}{v_jv_p}\right) k_j(v_j)^{-1} \right)
 +Lv_j\left(\frac{3v_j^2}{v_j^2-1}-\nu(v_j)\right), 
\end{gather*}
which are the usual form for the XXZ Gaudin model with boundary.

\section{Modif\/ied algebraic Bethe ansatz} \label{sec:maba}

In this section, we want to diagonalize the matrix $\frac{1}{4}\operatorname{tr}_0\big({\widetilde{\cK_0}}(x)^2\big)$ when the 4 parameters in the $k$-matrix are generic.
The diagonalisation when the $k$-matrix is diagonal was obtained in \cite{Hik} and for the triangular case in \cite{mano+,mano}.
For the generic case, the situation is more complicated. Here, we use the modif\/ied algebraic Bethe ansatz,
introduced and used in dif\/ferent contexts \cite{Bel++,Bel,BC,Bel+,cra}, to compute the eigenvalues and the eigenvectors for the generic case.
The f\/irst two subsections are technical. Then, Section~\ref{sec:tria} contains the known results about the diagonal or triangular boundaries
recovered here and Sections~\ref{sec:odd} and~\ref{sec:even} contains the new results for the chain of length odd or even.

\subsection{Algebraic relations \label{sec:alg}}

In this subsection, by using the commutation relations obtained from \eqref{eq:Al}, we provide necessary propositions to
perform the algebraic Bethe ansatz.
Firstly, let us mention that the commutation relations between $ A(x)$, $\widetilde B(x)$ and $\widetilde C(x)$ given by \eqref{eq:Al}
with the $r$-matrix \eqref{eq:rt} and
$\widetilde \cK(x)$ given by \eqref{eq:rep0} are quite complicated in comparison to the usual ones.
This is why we introduce the two following shifted generators
\begin{gather*}
 {B}(x,n)= \widetilde B(x)-(2n-1) b(x)\qquad\text{and}\qquad
 {C}(x,n)= \widetilde C(x)-(2n-1)c(x) ,
\end{gather*}
where $b(x)$ and $c(x)$ are def\/ined by \eqref{eq:defbc}. For these generators and the generators $ A(x)$, the commutation relations become
\begin{gather}
{A}(x){A}(y)={A}(y){A}(x),\label{eq:com1}\\
{B}(x,n){B}(y,n+1)={B}(y,n){B}(x,n+1),\label{eq:com2}\\
{C}(x,n){C}(y,n-1)={C}(y,n){C}(x,n-1),\label{eq:com3}\\
{A}(x){B}(y,m)={B}(y,m){A}(x)+\omega(x,y) {B}(y,m) -f(x,y){B}(x,m),\label{eq:com4}\\
{A}(x){C}(y,m)={C}(y,m){A}(x)-\omega(x,y) {C}(y,m) -f(1/x,1/y){C}(x,m),\label{eq:com5}\\
{C}(x,n){B}(y,n)={B}(y,n+1){C}(x,n+1) \nonumber\\
\hphantom{{C}(x,n){B}(y,n)=}{} +2f(x,y){A}(x) -2f(1/y,1/x){A}(y) - 8n b(y)c(x), \label{eq:com6}
\end{gather}
where $\omega(x,y)$ and $f(x,y)$ are def\/ined by \eqref{eq:defof}.

We introduce also a shifted transfer matrix
\begin{gather}\label{eq:deftau}
{t}(x,j)= \frac{1}{2}{A}(x)^2+\frac{1}{4}{B}(x,j+1){C}(x,j+1) +\frac{1}{4}{C}(x,j){B}(x,j)-\frac{1}{2}b(x)c(x).
\end{gather}
For $j=0$ in the previous relation, the transfer matrix \eqref{eq:tr} is recovered: $t(x)=\frac{1}{4} \operatorname{tr}_0 \widetilde{\cK_0}(x)^2 ={t}(x,0)$.

Secondly, we present a lemma permitting to commute the shifted transfer matrix with the generator $B(x,p)$.
This lemma is standard in the context of the algebraic Bethe ansatz for the Gaudin model except for the
modif\/ication of the shift for the transfer matrix.
\begin{Lemma}\label{lem:tB}
If the algebraic relations \eqref{eq:com1}--\eqref{eq:com6} are satisfied, the following relation holds
\begin{gather}
 {t}(x,p-1) {B}(z,p)- {B}(z,p){t}(x,p) = {B}(z,p)\big( \omega(x,z)(A(x)+\nu(x)) +2b(x)c(x) \big)\nonumber\\
\qquad{}-{B}(x,p)\big( f(1/z,1/x)(A(z)+\nu(z))+4(p-1)c(x)b(z) \big),\label{eq:tb}
\end{gather}
where $\nu(z)$ is defined by \eqref{eq:nu}.
\end{Lemma}

\begin{proof}
By straightforward computations using the commutation relations \eqref{eq:com1}--\eqref{eq:com6}, one gets
 \begin{gather*}
 {t}(x,p-1) {B}(z,p)- {B}(z,p){t}(x,p)= \omega(x,z) {B}(z,p) {A}(x) -f(1/z,1/x) {B}(x,p) {A}(z)\nonumber \\
 \qquad{}+\frac{1}{2}\big( \omega(x,z)^2+f(z,x)f(1/z,1/x) \big) {B}(z,p)\nonumber\\
 \qquad{}-\frac{1}{2}\big(\omega(z,x)f(1/z,1/x)+\omega(x,z) f(x,z) +4(2p-1)c(x)b(z) \big) {B}(x,p).
\end{gather*}
Then, by using the following functional relations
\begin{gather*}
 \omega(x,z)^2+f(z,x)f(1/z,1/x)=2\nu(x)\omega(x,z)+4b(x)c(x),\\
 \omega(z,x)f(1/z,1/x)+\omega(x,z) f(x,z)=2\nu(z)f(1/z,1/x)-4c(x)b(z)
\end{gather*}
one gets relation \eqref{eq:tb}.
\end{proof}

Let us def\/ine the following product
\begin{gather*}
\mathbb{B}(\boldsymbol{z})= {B}(z_1,1) {B}(z_2,2)\cdots {B}(z_M,M),
\end{gather*}
where $\boldsymbol{z}=\{z_1,z_2,\dots,z_M\}$.
Let us remark that, due to relation \eqref{eq:com2}, $\mathbb{B}(\boldsymbol{z})$ is invariant under any permutation of the $z_i$ which justif\/ies
the def\/inition of $\boldsymbol{z}$ as a set.
One introduces also, for $k=1,2,\dots,M$
\begin{gather*}
\mathbb{B}(\boldsymbol{z}_k,x)= {B}(z_1,1)\cdots {B}(z_{k-1},k-1) {B}(x,k) {B}(z_{k+1},k+1)\cdots {B}(z_M,M).
\end{gather*}
Finally, we are in position to compute the action of the transfer matrix with the product~$\mathbb{B}(\boldsymbol{z})$:
\begin{Proposition}\label{prop:1}
If the algebraic relations \eqref{eq:com1}--\eqref{eq:com6} are satisfied, the following relation holds
\begin{gather*}
t(x)\mathbb{B}(\boldsymbol{z}) = \frac{1}{2} \mathbb{B}(\boldsymbol{z})B(x,M+1)C(x,M+1)\\ 
\hphantom{t(x)\mathbb{B}(\boldsymbol{z}) =}{} +\mathbb{B}(\boldsymbol{z})\left(\Lambda(x)
+\sum_{p=1}^M\omega(x,z_p)\left(A(x)+\nu(x) +\sum_{q\neq p}\frac{(x-1/x)\omega(z_p,z_q)}{z_p-1/z_p}\right)\right)\nonumber\\
\hphantom{t(x)\mathbb{B}(\boldsymbol{z}) =}{} -\sum_{p=1}^M \mathbb{B}(\boldsymbol{z}_p,x)f(1/z_p,1/x)\left(A(z_p)+\nu(z_p) +\sum_{q\neq p}\omega(z_p,z_q)\right),\nonumber
\end{gather*}
where
\begin{gather*}
 \Lambda(x)=\frac{1}{2}A(x)^2+xA'(x)+\nu(x)A(x)-\frac{x^2+1}{x^2-1}A(x) -\frac{1}{2}b(x)c(x) .
\end{gather*}
\end{Proposition}

\begin{proof}
One gets, using a telescopic sum, the following relation
\begin{gather*}
{t}(x,0)\mathbb{B}(\boldsymbol{z})-\mathbb{B}(\boldsymbol{z}){t}(x,M) = \sum_{p=1}^M
 {B}_1\cdots {B}_{p-1}
\left( {t}(x,p-1) {B}_p- {B}_p{t}(x,p) \right)
 {B}_{p+1}\cdots {B}_M,
\end{gather*}
where $ {B}_p= {B}(z_p,p)$. Then, using the Lemma~\ref{lem:tB}, we express ${t}(x,p-1) {B}_p- {B}_p{t}(x,p)$
in the previous sum.
Then, we use repetitively relation \eqref{eq:com4} to put the operators $A$ on the right.
Then, we rearrange the terms using the following relations
\begin{gather*}
 \omega(x,z_p)f(x,z_q)-f(1/z_p,1/x)f(z_p,z_q)=\omega(z_q,z_p)f(1/z_q,1/x)-4c(x)b(z_q),\\
 \sum_{q>p}\omega(x,z_p)\omega(x,z_q)= \sum_{q\neq p} \frac{x-1/x}{z_p-1/z_p}\omega(z_p,z_q)\omega(x,z_p).
\end{gather*}
Finally, ${t}(x,M)$ is simplif\/ied using relation \eqref{eq:deftau} and \eqref{eq:com6} for $y\rightarrow x$ to prove
the proposition.
\end{proof}

Let us emphasize again that this type of computation and result are the same for the periodic or diagonal Gaudin model. The only dif\/ference lies on
the shift in the operators.

\subsection{Explicit representation} \label{sec:rep}

Up to now, we have only used the commutation relations of the algebra to prove the previous results.
Now, we choose the representation for $\widetilde{\cK}(x)$ given by \eqref{eq:rep}.
Therefore, we get explicitly
\begin{gather}
 A(x) = -\frac{1}{2}\sum_{j=1}^L\omega(x,v_j) \sigma^z_j,\nonumber\\
 B(x,n) = (1-2n) b(x)+b(x)\sum_{j=1}^L \sigma^z_j+\sum_{j=1}^L f(v_j,x) \sigma^-_j,\label{eq:repB}\\
 C(x,n) = (1-2n) c(x)+c(x) \sum_{j=1}^L \sigma^z_j -\sum_{j=1}^L f(1/v_j,1/x) \sigma^+_j.\nonumber
\end{gather}
Let us introduce the following vectors
$
 \Omega=\begin{pmatrix}
 1\\0
 \end{pmatrix}^{\otimes L}
$, called usually pseudo-vacuum. We deduce that{\samepage
\begin{gather}
 A(x)\Omega=a(x)\Omega \qquad \text{and}\qquad C(x,n)\Omega=(L+1-2n)c(x)\Omega\label{eq:rep3}
\end{gather}
with $a(x)=-\frac{1}{2}\sum\limits_{j=1}^L \omega(x,v_j)$.}

The Bethe vectors are def\/ined as follows
\begin{gather}\label{eq:BV}
 \mathbb{V}(\boldsymbol{z})=\mathbb{B}(\boldsymbol{z})\Omega.
\end{gather}
These Bethe vectors are of course invariant under any permutation of the $z_i$ since $\mathbb{B}(\boldsymbol{z})$ and $\mathbb{C}(\boldsymbol{z})$ are invariant.
From the previous results (Proposition \ref{prop:1} and relations \eqref{eq:rep3}), we get easily the action of the transfer matrix on these Bethe vectors
\begin{gather}
 t(x)\mathbb{V}(\boldsymbol{z}) = \frac{L-2M-1}{2}c(x)\mathbb{B}(\boldsymbol{z})B(x,M+1)\Omega\nonumber\\
\hphantom{t(x)\mathbb{V}(\boldsymbol{z}) =}{} + \left(\lambda(x)
+\sum_{p=1}^M\omega(x,z_p)\left(a(x)+\nu(x) +\sum_{q\neq p}\frac{(x-1/x)\omega(z_p,z_q)}{z_p-1/z_p}\right)\right)\mathbb{V}(\boldsymbol{z})\nonumber\\
\hphantom{t(x)\mathbb{V}(\boldsymbol{z}) =}{} -\sum_{p=1}^M f(1/z_p,1/x)\left(a(z_p)+\nu(z_p) +\sum_{q\neq p}\omega(z_p,z_q)\right)\mathbb{V}(\boldsymbol{z}_p,x),\label{eq:offsg}
\end{gather}
where
\begin{gather*}
 \lambda(x)=\frac{1}{2}a(x)^2+xa'(x)+\nu(x)a(x)-\frac{x^2+1}{x^2-1}a(x) -\frac{1}{2}b(x)c(x) .
\end{gather*}
The second term in \eqref{eq:offsg} is the wanted term (i.e., is proportional to the Bethe vector) and the third term is the unwanted term.
These terms are already present for the periodic or triangular case.
For the generic boundary, we must deal with the f\/irst term in the previous relation \eqref{eq:offsg}.
There are dif\/ferent possibilities:
\begin{itemize}\itemsep=0pt
 \item $c(x)=0$, i.e., $\rho=\pm \beta$. This case corresponds to triangular (or diagonal) boundaries and the Bethe ansatz works as usual.
 We get dif\/ferent sectors depending on the number $M$ of excitations (see Section \ref{sec:tria}).
 \item The length of chain is odd and $M=\frac{L-1}{2}$. It corresponds to a case where the number of excitations is f\/ixed. From relation \eqref{eq:offsg}, one gets only
 the half of the spectrum. The other half is obtained starting from $t(x)\overline{\mathbb{V}}(\boldsymbol{z})$ (see Section \ref{sec:odd}).
 \item The length of the chain is even. In this case, we must compute $\mathbb{B}(\boldsymbol{z})B(x,M+1)\Omega$.
 It corresponds to the modif\/ied algebraic Bethe ansatz (see Section \ref{sec:even}).
\end{itemize}

At f\/irst sight, the dif\/ference of the behavior for the lengths odd or even seems strange. However, we can show that, in the case of odd length,
two subspaces of $(\CC^2)^{\otimes L}$ are stabilized by the transfer matrix. More precisely, if we denote by $V_J$
the eigenspace of the total spin $S^z=\frac{1}{2}\sum\limits_{j=1}^L \sigma^z_j $
with eigenvalue $J$, one gets
\begin{gather}
 t(x) V_{1/2} \oplus V_{3/2}\oplus\dots \oplus V_{L/2} \subset V_{1/2} \oplus V_{3/2}\oplus\dots \oplus V_{L/2},\nonumber 
 \\
 t(x) V_{-1/2} \oplus V_{-3/2}\oplus\dots \oplus V_{-L/2} \subset V_{-1/2} \oplus V_{-3/2}\oplus\dots \oplus V_{-L/2}.\label{eq:sec2}
\end{gather}
Then it is natural that the spectrum for odd chain splits into two sectors.
For the chain of even length, this feature disappears: there is no stable subspace.

\subsection{Triangular and diagonal boundaries} \label{sec:tria}

In this subsection, we deal with the triangular boundary (i.e., $\rho=\pm\beta$).
In this case, we can easily see that one gets, for $M=0,1,\dots, L$,
\begin{gather*}
 t(x) V_{L/2-M} \oplus V_{L/2-M+1}\oplus\cdots \oplus V_{L/2} \subset V_{L/2-M} \oplus V_{L/2-M+1}\oplus\cdots \oplus V_{L/2}.
\end{gather*}
For $M=0,1,\dots,L$, the Bethe vectors $\mathbb{V}(\boldsymbol{z})$ def\/ined by \eqref{eq:BV} is just in the sector
$V_{L/2-M} \oplus\cdots \oplus V_{L/2}$ and becomes an eigenvector of the transfer matrix if the
 Bethe roots $z_i$ satisfy the Bethe equations:
\begin{gather}\label{eq:bethed}
 a(z_p)+\nu(z_p) +\sum_{\genfrac{}{}{0pt}{1}{q=1}{q\neq p}}^M\omega(z_p,z_q) =0 \qquad\text{for} \quad p=1,2,\dots,M.
\end{gather}
Explicitly, they read, for $p=1,\dots,M$,
\begin{gather}
-\frac{1}{2}\sum_{j=1}^L\left(\frac{z_p+v_j}{z_p-v_j}+\frac{z_pv_j+1}{z_pv_j-1}\right)
+\frac{1}{2}\left(\frac{\gamma z-\beta}{\gamma z+\beta}+\frac{\beta z-\gamma}{\beta z+\gamma} \right)\nonumber\\
\qquad {} +\sum_{\genfrac{}{}{0pt}{1}{q=1}{q\neq p}}^M\left(\frac{z_p+z_q}{z_p-z_q}+\frac{z_pz_q+1}{z_pz_q-1}\right)=0.\label{eq:bd}
\end{gather}
The corresponding eigenvalue of the transfer matrix $t(x)=\frac{1}{4}\operatorname{tr}_0\big(\widetilde \cK_0(x)^2\big)$ is
\begin{gather}\label{eq:vpd}
 \lambda(x)
+\sum_{p=1}^M\omega(x,z_p)\left(a(x)+\nu(x) +\sum_{\genfrac{}{}{0pt}{1}{q=1}{q\neq p}}^M\frac{(x-1/x)\omega(z_p,z_q)}{z_p-1/z_p}\right).
\end{gather}

Starting from the triangular boundary, we recover the diagonal one for $\alpha=0$.
We see that the eigenvalue \eqref{eq:vpd} and the Bethe equations \eqref{eq:bd}
do not depend on $\alpha$. Then, as shown previously for the Gaudin model \cite{mano+,mano}
(or for other models \cite{BCR13,pimenta}), the spectrum for triangular boundaries is similar to the one with diagonal boundaries.
Evidently, the Bethe eigenvectors depends on $\alpha$, via the function $b(x)$ present in $B(x,j)$ (see~\eqref{eq:repB}).
In particular, the function $b(x)$ vanishes for $\alpha=0$ and the Bethe vectors
$\mathbb{V}(\boldsymbol{z})$ are now in the sector $V_{L/2-M}$ which is consistent with
\begin{gather*}
t(x) V_{L/2-M} \subset V_{L/2-M}.
\end{gather*}
The last equation means that the total spin is conserved by the transfer matrix in the case of diagonal boundary.

\subsection{Generic boundaries for an odd chain} \label{sec:odd}

If the length of the chain is odd (i.e., there is an integer $\ell$ such that $L=2\ell+1$), we can chose the number of excitations as $M=\frac{L-1}{2}=\ell$
and the coef\/f\/icient in front of $\mathbb{B}(\boldsymbol{z})B(x,M+1)\Omega$ vanishes.
From the explicit expression \eqref{eq:repB} of $B(x,j)$, we see that
\begin{gather*}
\mathbb{V}(\boldsymbol{z})\in V_{1/2} \oplus V_{3/2}\oplus\dots \oplus V_{L/2} \qquad\text{for} \quad M=\ell,
\end{gather*}
which is one of the subspaces stabilized by the transfer matrix, as explained previously.
Then, the eigenvalues for this sector are given by \eqref{eq:vpd} with the Bethe equations \eqref{eq:bethed} for $M=\ell$.

The spectrum and the associated eigenvectors for the other sector \eqref{eq:sec2} are obtained from the Bethe vectors
\begin{gather*}
 \overline{\mathbb{V}}(\boldsymbol{z})= {C}(z_1,0) {C}(z_2,-1)\cdots {C}(z_\ell,-\ell+1) \begin{pmatrix}
 0\\1
 \end{pmatrix}^{\otimes L}.
\end{gather*}
It is easy to show that $\overline{\mathbb{V}}(\boldsymbol{z})$ belongs to the second sector.
A relation similar to \eqref{eq:offsg} by replacing ${\mathbb{V}}(\boldsymbol{z})$ by $\overline{\mathbb{V}}(\boldsymbol{z})$ can be proved.
Then, we show that, starting from $\overline{\mathbb{V}}(\boldsymbol{z})$, one gets exactly
the expression \eqref{eq:vpd} for the eigenvalues and \eqref{eq:bethed} for the Bethe equations.
Therefore, the spectrums in both sectors are the same.

Let us also mention that the separation into two sectors is not new in this context. Indeed, the complete spectrum of the XXZ spin chain with
certain constraints between the parameters of the boundaries is also obtained from two dif\/ferent Bethe vectors \cite{YanZ07}.

\subsection{Generic boundaries for an even chain} \label{sec:even}

For a chain with an even length, we must deal with the term $\mathbb{B}(\boldsymbol{z})B(x,M+1)\Omega$.
For a generic~$M$, this vector has no particular property but for $M=L$, special feature appears. This is the crucial point of
the modif\/ied algebraic Bethe ansatz.

Firstly, let us state the following lemma:
\begin{Lemma}\label{lem:M}
The entries of the matrix
 \begin{gather*}
 {\mathcal M}(z)=
 \prod_{q=1}^{\genfrac{}{}{0pt}{2}{\longrightarrow}{L}}\left(\frac{(z-v_q)(z-1/v_q)}{(z-z_q)(z-1/z_q)} B(z_q,q) \right)\\
 \hphantom{{\mathcal M}(z)=}{}\times\left(\frac{(z-1/z)}{(z-x)(z-1/x)} B(x,L+1) \right)
 \big( \delta(z) B(z,L+1)\big)^{-1}
\end{gather*}
have only simple poles at $z=0,\infty, z_1,z_2,\dots,z_{L}, 1/z_1,1/z_2,\dots,1/z_{L},x,1/x$.
\end{Lemma}
\begin{proof}
The proof of the existence of the poles of ${\mathcal M}(z)$ at $z=z_1,z_2,\dots,z_{L}, 1/z_1,1/z_2,\dots,1/z_{L}$ and $x$, $1/x$ is straightforward.
One gets
\begin{gather*}
\widehat B(z)=\delta(z) B(z,L+1)=\alpha(2L+1) -\alpha \sum_{j=1}^L \sigma^z_j+\sum_{j=1}^L \frac{2z(v_j-1/v_j)}{(z-v_j)(z-1/v_j)} \delta(v_j) \sigma^-_j .
\end{gather*}
This matrix $\widehat B(z)$ is a lower triangular matrix with the diagonal entries in the set $\{\alpha(L+1)$, $\alpha(L+3),\dots, \alpha(3L+1)\}$. In particular, this shows that it is invertible which justif\/ies the def\/inition of~${\mathcal M(z)}$. Due to the Cayley--Hamilton theorem, $\widehat B(z)^{-1}$ is a polynomial of~$\widehat B(z)$. Then, by using the fact that~$(\sigma^-)^2=0$, we prove that the poles at $z=v_j$ and $z=1/v_j$ of~$\widehat B(z)^{-1}$ remain simple. Then, ${\mathcal M}(z)$ has no pole at these points. In addition, its becomes a~non-singular diagonal matrix for $z=0$ which proves the simplicity of the pole of ${\mathcal M}(z)$ at $z=0$. The existence of the pole at $z=\infty$ is proven using the property $ {\mathcal M}(z)=-1/z^2 {\mathcal M}(1/z)$ which concludes the proof.
\end{proof}

Similar technical lemma has been proven and used in \cite{cra}. We are now in position to provide the following proposition:
\begin{Proposition}\label{pro:fon}
The following relation holds
\begin{gather*}
 \mathbb{B}(\boldsymbol z)B(x,L+1)\Omega=
 \widehat \lambda(x) \mathbb{V}(\boldsymbol{z})+
 \sum_{p=1}^L \frac{(z_p-1/z_p)\operatorname{Res}_{x=z_p}\big(\widehat \lambda(x)\big)}{(z_p-x)(z_p-1/x)} \mathbb{V}(\boldsymbol{z}_p,x)
\end{gather*}
 with $\boldsymbol z=\{z_1,z_2,\dots, z_L\}$ and
\begin{gather*}
\widehat \lambda(x)=\frac{\alpha(L+1)}{\delta(x)}
\prod_{q=1}^{L}\left(\frac{(x-v_q)(x-1/v_q)}{(x-z_q)(x-1/z_q)} \right).
\end{gather*}
\end{Proposition}
\begin{proof}
 The residues of $\mathcal M(z)\Omega$ are obtained by straightforward computations:
 \begin{gather*}
 \operatorname{Res}_{z=z_p}(\mathcal M(z)\Omega)=\frac{1}{\delta(z_p)}
 \prod_{\genfrac{}{}{0pt}{2}{q=1}{q\neq p}}^{L}
 \left(\frac{(z_p-v_q)(z_p-1/v_q)}{(z_p-z_q)(z_p-1/z_q)} \right) \left(\frac{(z_p-v_p)(z_p-1/v_p)}{(z_p-x)(z_p-1/x)} \right)
\mathbb{V}(\boldsymbol{z}_p,x),\\
\operatorname{Res}_{z=x}(\mathcal M(z)\Omega) = \frac{1}{\delta(x)}
\prod_{q=1}^{L}\left(\frac{(x-v_q)(x-1/v_q)}{(x-z_q)(x-1/z_q)} \right) \mathbb{V}(\boldsymbol{z}),\\
\operatorname{Res}_{z=0}(\mathcal M(z)\Omega) = -\frac{1}{\alpha(L+1)}\mathbb{B}(\boldsymbol{z})B(x,L+1)\Omega.
 \end{gather*}
By using $ {\mathcal M}(z)=-1/z^2 {\mathcal M}(1/z)$, we deduce that $\operatorname{Res}_{z=1/z_0}(\mathcal M(z)\Omega)=\operatorname{Res}_{z=z_0}(\mathcal M(z)\Omega)$ where $z_0=0$, $z_1$, $z_2,\dots, z_L$, $x$. The proof is concluded by using the fact that the sum over all the residues (with the point at inf\/inity) of a rational function vanishes.
\end{proof}

Therefore, by using relation \eqref{eq:offsg} for $M=L$ and the result of Proposition \ref{pro:fon}, we prove that $\mathbb{V}(\boldsymbol z)$
(for $\boldsymbol z=\{z_1,z_2,\dots, z_L\}$) is an eigenvector of the transfer matrix $t(x)$ if the Bethe roots satisfy the following Bethe equations
\begin{eqnarray}\label{eq:bg}
a(z_p)+\nu(z_p) +\sum_{\genfrac{}{}{0pt}{1}{q=1}{q\neq p}}^L\omega(z_p,z_q)=\frac{(L+1)c(z_p) }{4 z_p} \operatorname{Res}_{x=z_p}\big(\widehat \lambda(x)\big),
\qquad\text{for }p=1,2,\dots,L.
\end{eqnarray}
To get the previous relation, we have used the functional relation
\begin{gather*}
-\frac{(L+1)(z-1/z)c(x)}{2f(1/z,1/x)(z-x)(z-1/x)}=\frac{(L+1)c(z) }{4 z}.
\end{gather*}
Explicitly, they become
\begin{gather*}
 -\frac{1}{2}\sum_{j=1}^L\left(\frac{z_p+v_j}{z_p-v_j}+\frac{z_pv_j+1}{z_pv_j-1}\right)
+\nu(z_p)+\sum_{\genfrac{}{}{0pt}{1}{q=1}{q\neq p}}^L\left(\frac{z_p+z_q}{z_p-z_q}+\frac{z_pz_q+1}{z_pz_q-1}\right)
\nonumber\\
\qquad{}=
\frac{(\rho^2-\beta^2)(L+1)^2}{16\rho^2} \frac{(z_p-v_p)(z_p-1/v_p)}{\delta(z_p)\delta(1/z_p)(z_p^2-1)}\prod_{\genfrac{}{}{0pt}{1}{q=1}{q\neq p}}^L
\frac{(z_p-v_q)(z_p-1/v_q)}{(z_p-z_q)(z_p-1/z_q)},
\end{gather*}
where $\nu(z)$ is given by \eqref{eq:nu}.
The corresponding eigenvalue is
\begin{gather}\label{eq:eig}
 \lambda(x)-\frac{L+1}{2}c(x)\widehat\lambda(x)
+\sum_{p=1}^M\omega(x,z_p)\left(a(x)+\nu(x) +\sum_{q\neq p}\frac{(x-1/x)\omega(z_p,z_q)}{z_p-1/z_p}\right).
\end{gather}
Then the l.h.s.\ of the Bethe equations is similar to the ones for the diagonal or triangular boundaries. The r.h.s.\ corresponds to the inhomogeneous part of the TQ relation discovered in~\cite{CYSW1,CYSW1+} to solve the XXZ spin chain with twist or boundaries.

\section{Higher spin Gaudin models} \label{sec:HS}

In this section, we generalize the previous construction and consider the higher spin Gaudin models. Indeed, it is well-known that we can f\/ind representations of the algebra $\cR$ using dif\/ferent representations of $\mathfrak{sl}(2)$ at each sites. Indeed, the following matrix $\widetilde{\cK}(x)$ satisf\/ies the def\/ining relation \eqref{eq:Al}:
 \begin{gather}\label{eq:rep2}
 \widetilde{\cK}_0(x)=\sum_{j=1}^L \cL^{(s_j)}_{0j}(x,v_j),
\end{gather}
where $s_j=1/2,1,3/2,\dots $ is the spin of the representation at the sites $j$ and
\begin{gather}\label{eq:Ls}
 \cL^{(s)}(x,y)=\begin{pmatrix}
 - \omega(x,y) S^z& 2b(x) S^z+ f(y,x) S^-\\
 2c(x) S^z
 - f(1/y,1/x) S^+&\omega(x,y)S^z
 \end{pmatrix}.
\end{gather}
In \eqref{eq:Ls}, $S^z$, $S^+$ and $S^-$ are the usual matrices representing the $\mathfrak{sl}(2)$ generators for spin $s$. The only necessary knowledge
about these matrices for the following constructions is the fact that there exists a highest-weight vector $w_s$ such that
\begin{gather*}
 S^z w_s = s w_s\qquad \text{and}\qquad S^+ w_s = 0 \qquad\text{for}\quad s=1/2,1,3/2,\dots
\end{gather*}
For $s_1,s_2, \dots, s_L=\frac{1}{2}$, the representation \eqref{eq:rep2} becomes \eqref{eq:rep}. Let us emphasize that the representation of $\mathfrak{sl}(2)$
may be dif\/ferent at each sites.

Now, we can generalize the results of Section \ref{sec:maba}. All the results of Section \ref{sec:alg} remains unchanged since it is independent of the representation.

 The pseudo-vacuum becomes\footnote{In this section, we keep the same notations for all the objects which generalize the objects of the previous sections.}
 \begin{gather*}
 \Omega=w_{s_1} \otimes w_{s_2} \otimes \dots \otimes w_{s_L}.
 \end{gather*}
 Then, relation \eqref{eq:offsg} is formally unchanged except the def\/inition of the function $a(x)$ which becomes
 \begin{gather}\label{an}
 a(x)=-\sum_{j=1}^L \omega(x,v_j)s_j
 \end{gather}
 and the $L$ in the f\/irst line must be replaced by $\overline{L}=2\sum\limits_{p=1}^L s_p$.

 Then, there is again three possibilities to pursue the computation:
 \begin{itemize}\itemsep=0pt
 \item The triangular case corresponding to $c(x)=0$. The Bethe equations and the eigenvalues are given respectively by \eqref{eq:bethed} and \eqref{eq:vpd}
 with $a(x)$ given by \eqref{an}. The number $M$ ranges now from 0 to $\overline{L}$.
 \item The odd case $M=\frac{\overline{L}-1}{2}$. The discussion of Section \ref{sec:odd} is still valid.
 \item The even case. To deal with this case, we must generalize the Lemma \ref{lem:M}.
 This lemma is still valid if we def\/ine ${\mathcal M}(z)$ as follows
 \begin{gather*}
 {\mathcal M}(z) =
\prod_{q=1}^L (z-v_q)^{2s_q}(z-1/v_q)^{2s_q}
\prod_{q=1}^{\genfrac{}{}{0pt}{2}{\longrightarrow}{\overline{L}}} \frac{B(z_q,q)}{(z-z_q)(z-1/z_q)} \nonumber\\
\hphantom{{\mathcal M}(z) =}{} \times \frac{(z-1/z)}{(z-x)(z-1/x)} B(x,\overline{L}+1)
 \big( \delta(z) B(z,\overline{L}+1)\big)^{-1} .
\end{gather*}
Let us emphasize that the f\/irst product is up to $L$, the number of sites, and the second up to $\overline{L}$. The demonstration is very
similar to the one given previously except we must use $(S^-)^{2s+1}= 0$ for a spin $s$ representation. Then, the Bethe equations and
the eigenvalues are given respectively by \eqref{eq:bg} and \eqref{eq:eig} with $L$ replace by $\overline{L}$, $a(z)$ given by \eqref{an} and
\begin{gather*}
\widehat \lambda(x)=\frac{\alpha(\overline{L}+1)}{\delta(x)}
\prod_{q=1}^L (x-v_q)^{2s_q}(x-1/v_q)^{2s_q}\prod_{q=1}^{\overline{L}} \frac{1}{(x-z_q)(x-1/z_q)}.
\end{gather*}
 \end{itemize}

\subsection*{Acknowledgements}

I thank P.~Baseilhac, S.~Belliard and V.~Caudrelier for their interest.
This work has been done during the stay of the author at the ``Laboratoire de Math\'ematiques et Physique Th\'eorique CNRS/UMR 7350,
Universit\'e de Tours''. I thank warmly the LMPT for hospitality. 

\pdfbookmark[1]{References}{ref}
\LastPageEnding

\end{document}